\documentclass[a4paper,12pt]{article}
\usepackage{amsmath,amsthm,amsfonts,amssymb,bbm}
\usepackage{graphicx,psfrag,subfigure,color,cite}

\numberwithin{equation}{section}

\newcommand{\e}{\varepsilon}
\newcommand{\Pb}{\mathbb{P}}

\newcommand{\I}{\mathrm{i}}
\newcommand{\dx}{\mathrm{d}}
\newcommand{\R}{\mathbb{R}}
\newcommand{\N}{\mathbb{N}}
\newcommand{\Z}{\mathbb{Z}}

\newcommand{\Id}{\mathbbm{1}}
\newcommand{\Ai}{\mathrm{Ai}}
\renewcommand{\Re}{\operatorname{Re}}

\newcommand{\Or}{{\cal O}}

\newtheorem{prop}{Proposition}[section]
\newtheorem{thm}[prop]{Theorem}
\newtheorem{lem}[prop]{Lemma}
\newtheorem{defin}[prop]{Definition}

\newtheorem{cla}[prop]{Claim}

\newtheorem{rem}[prop]{Remark}
\newenvironment{remark}{\begin{rem}\normalfont}{\end{rem}}

\title{Shock fluctuations in flat TASEP under critical scaling}
\author{Patrik L.\ Ferrari\thanks{Institute for Applied Mathematics, Bonn University, Endenicher Allee 60, 53115 Bonn, Germany. E-mail: {\tt ferrari@uni-bonn.de}} \and
Peter Nejjar\thanks{Institute for Applied Mathematics, Bonn University, Endenicher Allee 60, 53115 Bonn, Germany. E-mail: {\tt nejjar@uni-bonn.de}}
}

\date{August 21, 2014}

\begin{document}
\maketitle
\sloppy

\begin{abstract}
We consider TASEP with two types of particles starting at every second site. Particles to the left of the origin have jump rate $1$, while particles to the right have jump rate $\alpha$. When $\alpha<1$ there is a formation of a shock where the density jumps to $(1-\alpha)/2$. For $\alpha<1$ fixed, the statistics of the associated height functions around the shock is asymptotically (as time $t\to\infty$) a maximum of two independent random variables as shown in~\cite{FN14}.
In this paper we consider the critical scaling when $1-\alpha=a t^{-1/3}$, where $t\gg 1$ is the observation time. In that case the decoupling does not occur anymore. We determine the limiting distributions of the shock and numerically study its convergence as a function of $a$. We see that the convergence to $F_{\rm GOE}^2$ occurs quite rapidly as $a$ increases. The critical scaling is analogue to the one used in the last passage percolation to obtain the BBP transition processes~\cite{BBP06}.
\end{abstract}

\newpage

\section{Introduction}\label{SectIntro}
We consider the totally asymmetric simple exclusion process (TASEP) on $\Z$ with particle-dependent rates. Each site can be occupied by at most a particle (exclusion constraint). Particles  try to jump to their neighboring site with a given rate and the jump occurs provided the target site is empty.

Depending on the initial conditions and / or choice of the jump rates, the macroscopic density of particle can have discontinuities, also called shocks. For particle-independent jump rates, it is known that the stationary measures are either blocking measures (with no current) or product Bernoulli measures with fixed density~\cite{Li99}. In this setting, naturally one considered random initial conditions with Bernoulli initial conditions but density $\rho_-$ to the left and $ \rho_+$ to the right of the origin with the condition $0<\rho_-<\rho_+<1$. The shock moves with speed $1-\rho_- - \rho_+$ and has Gaussian fluctuations in the scale $t^{1/2}$ for large time $t$~\cite{Fer90,FF94,PG90} (the shock position is often defined by the position of a second class particle). The same holds for the fluctuations of a tagged particle around the macroscopic shock position.

Shocks are located where the characteristics of the associated PDE meet. The Gaussian fluctuations and the $t^{1/2}$ scale are due to the randomness in the initial condition (see~\cite{vanB91} for a physical argument). The reason is that at the particles at the shock are non-trivially correlated with two regions at time $0$ which are of order $t$ away from the origin (specifically, the positions where the characteristics meeting at the shock start), whose fluctuations are Gaussian and of order $t^{1/2}$. Since TASEP belongs to the Kardar-Parisi-Zhang (KPZ) universality class, the fluctuations created by the dynamics until time $t$ are  $\mathcal{O}\left(t^{1/3}\right)$   and therefore are not seen in the $t^{1/2}$ scale.

In our previous work~\cite{FN14} we analyzed one situation where shocks occur with non-random initial condition. In that case, the shock width is of order $t^{1/3}$ and in that scale the fluctuations of a tagged particle has a law that is the maximum of two random variables, because the pieces of information coming from each of the two characteristics that meet at the shock are asymptotically independent.

In this paper we consider the situation where at time $t=0$ particles occupy the whole $2\Z$, particles starting from $2\Z_-$ have unit jump rate and the ones starting from $2\Z_+$ have jump rate $\alpha$. We consider large time $t$ and the critical scaling $\alpha-1=\Or(t^{-1/3})$. In this situation the pieces of  information from the two characteristics are strongly correlated and, at the same time, the system differs from the constant density case. Therefore we call this a critical scaling. We obtain the limiting process describing the fluctuations of the particles positions around the shock (see Theorem~\ref{ThmMain}). We then perform a numerical study of the convergence of the distribution function to the one expected in the case of a shock with density jump of order one obtained in~\cite{FN14}. For the numerical evaluation we use the Matlab program developed by Bornemann~\cite{Born08} and we see that the convergence is surprisingly fast, see Figure~\ref{FigDistr}.

As discussed in Section~\ref{SectLPP}, a similar critical scaling, where particles starts from $\Z_-$ but the first $n$ particles have jump rate $\alpha$ has been considered in the context of last passage percolation. In the large time $t$ limit, the distribution function of a particle that is around the origin at time $t$ has a BBP distribution function~\cite{BBP06,BP07}.

The rest of the paper is organized as follows. In Section~\ref{SectMainResults} we define the model, discuss the macroscopic behavior, provide the main result, and finally present the numerical results. We also present the result  in terms of a last passage percolation model. In Section~\ref{SectAsympt} we provide the proof of the main theorem. In the appendix we give  an explicit formula in terms of Airy function of the kernel appearing in the main theorem.

\subsubsection*{Acknowledgments} P.L. Ferrari was supported by the German Research Foundation via the SFB 1060--B04 project. P. Nejjar is grateful for the support of the Bonn International Graduate School (BIGS).

\section{Model and results}\label{SectMainResults}
We consider the totally asymmetric simple exclusion process and assign to each particle a label. We denote by $x_n(t)\in\Z$ the position of particle with label $n$ at time $t$. The dynamics is as follows: each particle independently tries to jump to its right site and the jump occurs if that site is empty. The rate with which particle $n$ tries to jump from $x_n(t)$ to $x_n(t)+1$ is denoted by $v_n$. The dynamics preserves the order of particle and we use the convention to label from right to left, i.e., $x_n(t)>x_{n+1}(t)$ for any $n\in\Z$ and time $t$. The process is well-defined since in a finite time no information is coming from $\pm\infty$ (as one can see either from the graphical construction~\cite{Har78} or from the general Hille-Yoshida semigroup approach for interacting particle systems, see Section I.3-4 of~\cite{Li85b}).

In this paper we consider the following specialization of the model:
\begin{align}\label{twospeedTASEP}
x_{n}(0)=-2n, n\in \mathbb{Z},\quad
v_n=\left\{\begin{array}{ll} 1,& n>0, \\\alpha, & n\leq 0. \end{array}\right.
\end{align}
As there are two jump rates involved, we call this model \textit{two-speed TASEP}. Further, we call particles with label $n\leq 0$ \emph{$\alpha$-particles} and the ones with label $n>0$ are called \emph{normal particles}. We will be interested in the fluctuations of particle positions of normal particles which are around the shock created by the $\alpha$-particles. First we discuss the macroscopic picture.

\subsection{Macroscopic behavior}
In the main result we describe the fluctuations with respect to the deterministic macroscopic behavior, given as follows. Under hydrodynamic scaling, the evolution of the particle density $\varrho$ of the normal particles is governed by the Burgers equation
\begin{equation}\label{eq1}
\partial_t \varrho+\partial_x \varrho(1-\varrho)=0.
\end{equation}
The first normal particle is, however blocked by the last $\alpha$-particle, which starts at the origin and moves with an average speed $\alpha/2$. Therefore, the macroscopic density profile of the normal particles can be obtained by solving \eqref{eq1} for $x\in(-\infty,t\alpha/2]$ under the boundary condition \mbox{$\varrho(t \alpha/2,t)=\min\{\alpha/2,1\}$}. The system we consider corresponds to the initial condition $\varrho(x,0)=1/2$ for $x\leq 0$ and the macroscopic density profile of the normal particles is as follows: for $\alpha\geq 1$, there is a rarefaction fan given by
\begin{equation}
\varrho(\xi t,t)=\left\{\begin{array}{ll}
1/2, &\textrm{for }\xi\leq 0,\\
(1-\xi)/2,&\textrm{for }\xi\in [0,\min\{1,\alpha-1\}],\\
\max\{0,1-\alpha/2\},&\textrm{for }\xi\in [\min\{1,\alpha-1\},\alpha/2],
\end{array}\right.
\end{equation}
while for $\alpha\in [0,1)$ there is a shock moving with speed $(\alpha-1)/2$,
\begin{equation}
\varrho(\xi t,t)=\left\{\begin{array}{ll}
1/2, &\textrm{for }\xi\leq(\alpha-1)/2,\\
1-\alpha/2,&\textrm{for }\xi\in [(\alpha-1)/2,\alpha/2].
\end{array}\right.
\end{equation}
As a consequence, the particle with label $n=\lfloor \eta t\rfloor$ will be around the shock position if $\eta\simeq (2-\alpha)/4$.

\subsection{Critical scaling regime}
In this paper we focus at the critical regime where the discontinuity of the density is small. For a fixed $a\in\R$ and a large observation time $t\gg 1$, we scale the jump rate of the $\alpha$-particles critically, i.e., we consider
\begin{equation}\label{alphascaling}
\alpha=1-2a (t/2)^{-1/3}.
\end{equation}
In view of the macroscopic description, particles with number given by \mbox{$(2-\alpha) t/4=t/4+a (t/2)^{2/3}$} is around the shock and its position is around $-2n+t/2=-2a(t/2)^{2/3}$. Unlike in the macroscopic shock studied in~\cite{FN14}, where around the shock position particles are correlated over a $t^{1/3}$ scale, in the present situation particles are correlated over a $t^{2/3}$ scale. Therefore we consider the scaling\footnote{In what follows we will not write the integer parts explicitly.}
\begin{equation}\label{eq2.6}
n(u,t)=\left\lfloor \frac{t}{4}+(a+u)(t/2)^{2/3}\right\rfloor,\quad x(u,t)=\left\lfloor-2(a+u)(t/2)^{2/3}\right\rfloor,
\end{equation}
and define the accordingly scaled particle position process by
\begin{equation}
u\mapsto X_t(u)=\frac{x_{n(u,t)}-x(u,t)}{-(t/2)^{1/3}}.
\end{equation}

Our main analytic result is the limit process \mbox{${\cal M}_a=\lim_{t\to\infty} X_t$}.
\begin{defin}[The limit process $\mathcal{M}_a$]\label{defLimitProcess}
Define the extended kernel
\begin{equation}
\begin{aligned}
&K_{a}(u_1,\xi_1;u_2,\xi_2)=-\frac{1}{\sqrt{4\pi(u_2-u_1)}}\exp\left(-\frac{(\xi_2-\xi_1)^{2}}{4(u_2-u_1)}\right)\Id(u_2>u_1)
\\&+\frac{-1}{(2\pi\I)^{2}}\int_{\gamma_{+}}\dx w\int_{\gamma_{-}}\dx z
\frac{e^{w^{3}/3+(u_2+a)w^{2}-\xi_2 w}}{e^{z^{3}/3+(u_1+a)z^{2}-\xi_1 z}}\frac{2w}{(z-w)(z+w)}
\\&+\frac{1}{(2\pi\I)^{2}}\int_{\Gamma_{+}} \dx w\int_{\Gamma_-} \dx z
\frac{e^{w^3/3+(u_2-a) w^2-w(\xi_2+4a u_2)+4 u_2 a^2}}{e^{z^3/3+(u_1-a) z^2-z(\xi_1+4a u_1)+4 u_1 a^2}}\frac{2(w-2a) }{(z+w)(z-w+4a)}.
\end{aligned}
\end{equation}
The curves can be chosen as follows. Let \mbox{$\theta=\max\{|u_1|+|a|,|u_2|+|a|\}$}. For any choice of $r_\pm,R_\pm$ satisfying $r_+> -r_-> \theta$ and \mbox{$-R_->R_+>\theta+4|a|$}, we can set $\gamma_\pm=r_\pm+\I\R$ and $\Gamma_\pm=R_\pm+\I\R$ (oriented with increasing imaginary parts). The limit process $\mathcal{M}_{a}$ is defined by its finite-dimensional distribution: for any given $u_1<u_2<\cdots<u_m$,
 \begin{equation}\label{eq2.9}
 \Pb\left(\bigcap_{k=1}^{m}\{\mathcal{M}_{a}(u_k)\leq s_k\}\right)
 =\det(\Id-\chi_{s} K_a \chi_s)_{L^{2}(\{u_1,\ldots,u_m\}\times\mathbb{R})}
 \end{equation}
 where $\chi_{s}(u_k, x)=\Id(x>s_k)$. An explicit expression of $K_a$ in terms of Airy functions is given in Appendix~\ref{SectAppendix}.
\end{defin}

With this definition we can state our main analytic result, proven in Section~\ref{SectAsympt}.
\begin{thm}\label{ThmMain}
It holds
\begin{equation}
\lim_{t\to \infty}X_{t}(u)=\mathcal{M}_{a}(u)
\end{equation}
in the sense of finite dimensional distributions.
\end{thm}

\begin{remark}
In some special cases or limits we recover previous known processes. For example:
\begin{itemize}
\item[(a)] For $a=0$ we have the flat TASEP and see the Airy$_1$ process~\cite{BFPS06,Sas05}: ${\cal M}_0(u)=2^{1/3}{\cal A}_1(2^{-2/3} u)$.
\item[(b)] When $a\to -\infty$ a rarefaction fan is created. At his left edge, ${\cal M}_a$ becomes the Airy$_{2\to 1}$ process~\cite{BFS07}: $\lim_{a\to -\infty}{\cal M}_a(u-a) = {\cal A}_{2\to 1}(u)$. Inside the rarefaction fan ${\cal M}_a$ becomes the Airy$_2$ process~\cite{PS02}. For instance, in the middle of the rarefaction fan: $\lim_{a\to -\infty}{\cal M}_a(u)+(u+a)^2 = {\cal A}_{2}(u)$.
\item[(c)] For $a>0$, there is a shock and ${\cal M}_a$ is a transition process between two ${\cal A}_1$ processes. Indeed, $\lim_{M\to\infty}{\cal M}_a(u\pm M) = 2^{1/3}{\cal A}_{1}(2^{-2/3} u)$.
\end{itemize}
\end{remark}
Further, when $a\to\infty$ one should recover the macroscopic shock picture and by the result of~\cite{FN14} the one-point distribution should become a product of two $F_1$ distributions, where $F_1$ is the GOE Tracy-Widom distribution discovered first in random matrix theory~\cite{TW96}. More precisely, extrapolating the result of Corollary~2.5 in~\cite{FN14} (reported below for convenience) we conjecture that
\begin{equation}\label{eqConj}
\lim_{a\to\infty} \Pb({\cal M}_a(0)\leq s)=(F_1(2^{2/3}s))^2.
\end{equation}
As the kernel $K_a$ is not converging pointwise as $a\to\infty$ (not even after an appropriate conjugation) we could not verify the conjecture analytically. The numerical studies presented in Section~\ref{SectNumerics} below are in agreement with the conjecture and, moreover, show that the convergence as $a$ increases is quite fast.

\begin{thm}[Corollary 2.5 in~\cite{FN14}]
Let $x_n(0)=-2n$ for $n\in\Z$. For $\alpha<1$ let $\eta=\frac{2-\alpha}{4}$ and $v=-\frac{1-\alpha}{2}$. Then it holds
\begin{equation}\label{eqCor1b}
\lim_{t\to\infty}\Pb\left(x_{\eta t}(t)\geq v t - s t^{1/3}\right)=F_1\left(2s\right)F_1\left(2s\sigma_\alpha\right),
\end{equation}
with $\sigma_\alpha=\frac{\alpha^{1/3}(2-2\alpha+\alpha^2)^{1/3}}{(2-\alpha)^{2/3}}$. Note that $\sigma_{\alpha}\to 1$ as $\alpha\to 1$.
\end{thm}

\subsection{Last passage percolation}\label{SectLPP}
The limit process ${\cal M}_a$ occurs in a related last passage percolation (LPP) model as well. To each site $(i,j)$ of $\mathbb{Z}^{2}$ we assign an independent random variable $\omega_{(i,j)}$ with
\begin{equation}
\omega_{(i,j)}\sim \exp(v_j).
\end{equation}
Further, to a TASEP initial condition $\{x_{n}(0),n\in\Z\}$ we assign the line
\begin{equation}
{\cal L}=\{(n+x_{n}(0),n)\,|\, n\in\Z\}.
\end{equation}
For a point $(m,n)$ on the right/above the line $\cal L$, the last passage time from ${\cal L}$ to $(m,n)$ is defined by
\begin{equation}
L_{{\cal L}\to (m,n)}:=\max_{\pi:{\cal L} \to (m,n)} \sum_{(i,j)\in \pi}\omega_{(i,j)},
\end{equation}
where the maximum is taken over all up-right paths\footnote{An up-right path $\pi=\{\pi(0),\ldots,\pi(n)\}$  is a sequence of points of $\Z^2$ such that \mbox{$\pi(i+1)-\pi(i)\in\{(1,0),(0,1)\}$}, for $i=0,\ldots,n-1$.} from $\cal L$ to $(m,n)$.
The well-known connection between TASEP and LPP is
\begin{equation}\label{eqLPPtasep}
\Pb\bigg(\bigcap_{k=1}^{r}\{ x_{n_{k}}(t)\geq m_{k}-n_{k}\}\bigg)=\Pb\bigg(\bigcap_{k=1}^r \{L_{{\cal L}\to (m_k,n_k)}\leq t\}\bigg).
\end{equation}
In our model, we have ${\cal L}=\{(-n,n), n\in\Z\}$. Consider the critical scaling
\begin{equation}
\alpha=1-2a(2\ell)^{-1/3},
\end{equation}
and focus at the position
\begin{equation}
m(v,\ell)=\ell-2(v+a)(2\ell)^{2/3},\quad n=\ell.
\end{equation}
Define the rescaled LPP time by
\begin{equation}
L^{\rm resc}_\ell(v):=\frac{L_{{\cal L}\to (m(v,\ell),\ell)}-\left[4 \ell - 4 (v+a)(2\ell)^{2/3}\right]}{2(2\ell)^{1/3}}.
\end{equation}
\begin{thm}\label{ThmMainLPP}
It holds
\begin{equation}
\lim_{\ell\to \infty}L^{\rm resc}_\ell(v)=\mathcal{M}_{a}(v)
\end{equation}
in the sense of finite dimensional distributions.
\end{thm}
In short, to prove Theorem~\ref{ThmMainLPP}, one starts with the relation (\ref{eqLPPtasep}) that gives the joint distributions of $L^{\rm resc}_\ell$ in terms of positions of TASEP particles at different times, varying around $t=4\ell$ on a $\ell^{2/3}$ scale only. By the slow-decorrelation phenomenon~\cite{Fer08,CFP10b} the fluctuations at different times are asymptotically the same as the fixed time fluctuations for points lying on special space-time directions (the characteristics). At fixed time, the result is exactly given by Theorem~\ref{ThmMain}. The details of this procedure have been worked out for instance in~\cite{BFP09,CFP10a}.

\subsection{Numerical study}\label{SectNumerics}
Here we numerically compute the distribution function of the process ${\cal M}_a$, given by a Fredholm determiant of the kernel $K_a$, as well as some of its basic statistics: Expectation, Variance, Skewness, Kurtosis. For the computation, we use the formula for $K_a$ given in the Appendix~\ref{SectAppendix}, since Airy functions are already implemented functions in Matlab, and apply Bornemann's method for the evaluation of the Fredholm determinants, see~\cite{Born10}, which is well-adapted for analytic kernels. Bornemann's algorithm also comes with an error control that we used, see Section 4.4 of~\cite{Born10}.
 \begin{figure}
 \begin{center}
 \includegraphics[height=6cm]{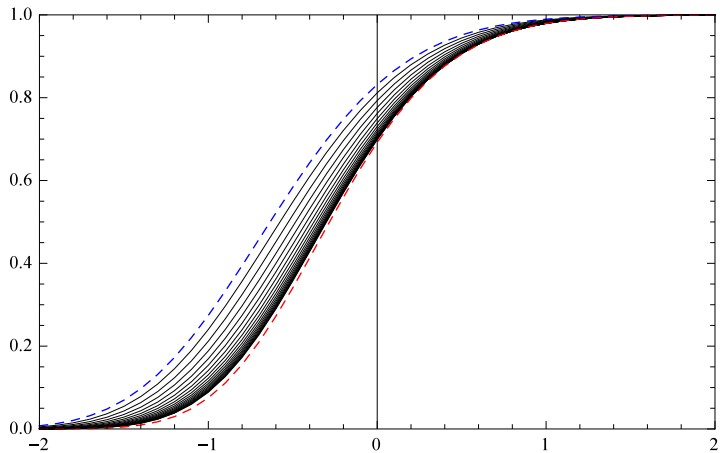}
 \caption{The dashed lines are $s\mapsto F_{1}(2s)$ (the left one) and \mbox{$s\mapsto F_{1}(2s)^{2}$} (the right one), $s\in [-2,2]$. The solid  lines are the functions $s\mapsto G_{a}(s)$ for \mbox{$a=0.1,0.2,\ldots,1.8$}. For $a=0$, $G_0(s)=F_1(2s)$, when $a$ grows larger, $G_a(s)$ approximates the macroscopic shock distribution $F_{1}(2s)^{2}$ as conjectured, see~\eqref{eqConjEquiv}.}
\label{FigDistr}
 \end{center}
 \end{figure}

For simplicity, we  study the validity of the conjecture in \eqref{eqConj}, i.e., we set $u=0$. In principle, one could look also general $u$, but then it has to be taken as a function of $a$ too (since in the unscaled process the correlation scale of the process changes from $t^{2/3}$ to $t^{1/3}$ as $\alpha$ varies from $1$ to a value strictly less then $1$).

To avoid to carry around a lot of $2^{1/3}$ factors, we rescale space by a factor $2^{1/3}$ so that the conjecture \eqref{eqConj} writes
\begin{equation}\label{eqConjEquiv}
\lim_{a\to\infty} \Pb({\cal M}_a(0)\leq s 2^{1/3})=(F_1(2s))^2.
\end{equation}
We denote $\widetilde{K}_{a}(\xi_1,\xi_2):=2^{1/3}K_{a}(0,2^{1/3}\xi_1;0,2^{1/3}\xi_2)$. Remark that the special case $a=0$ we have the Airy$_1$ kernel, $\widetilde{K}_{0}(\xi_1,\xi_2)=\Ai(\xi_1+\xi_2)$. By (\ref{eq2.9}) we have
\begin{align}
G_{a}(s):=\Pb({\cal M}_a(0)\leq s 2^{1/3})=\det(\Id-\chi_{s}^{c} \widetilde{K}_a \chi_{s}^{c}),\quad \chi_{s}^{c} =\Id_{(s,\infty)}.
\end{align}

By Theorem~\ref{ThmMain}, this is the $t\to\infty$ limit of the rescaled position of a particle in the microscopic shock. As mentioned earlier, we let $a$ grow large so as to recover the macroscopic shock distribution, which for large but finite time $t$ would correspond to the choice $a=\frac{(1-\alpha)t^{1/3}}{2^{4/3}}$. Due to the numerical limitations discussed below, we will compute $G_a$ and its basic statistics up to $a=1.8$. Surprisingly, already for this relatively small value of $a$, one is already quite close to the asymptotic behavior. The reason for this is the following. At first approximation, from the KPZ scalings, we know that the randomness that influences the statistical properties of particle positions around the shock lives in a $t^{2/3}$ neighborhood of the characteristic lines that comes together at the shock (for a proof in a special case, see~\cite{Jo00b}) and the neighborhood should be quite tight to provide the super-exponential decay of the covariance for the Airy$_1$ process~\cite{BFP08}. Further, by a closer inspection near the end-points, we discover that at $t^{1/3}$ distance from the shock, the neighborhood is only of order $t^{1/3}$ as well~\cite{FN14,FS03b}. These two phenomena imply that the convergence will happen on $a$ of order $1$.

\subsubsection*{Numerical Limitations} The limitation to $a\leq 1.8$ is due to the numerical difficulty of evaluating $G_a$ for $a$ large. As $a$ grows, $\widetilde{K}_a$ has some terms which are of order $1$ and one term which is (super-) exponentially diverging. More precisely, one has
\begin{equation}
\begin{aligned}
&\eqref{divstrong}=e^{4a^{3}/3-2a(\xi_1+\xi_2)-(\xi_2-\xi_1)^{2}/16a-\ln(4\sqrt{\pi a})}+\varepsilon_{0}(a)
\\&\eqref{div1}=2^{-1/3}\mathrm{Ai}(2^{-1/3}(\xi_1+\xi_2))e^{-a(\xi_2-\xi_1)}+\varepsilon_{1}(a)
\\&\eqref{div2}=2^{-1/3}\mathrm{Ai}(2^{-1/3}(\xi_1+\xi_2))e^{-a(\xi_1-\xi_2)}+\varepsilon_{2}(a)
\\&
\left|\eqref{vanishing}\right|\leq c_{i}\max_{\lambda\geq 0} \mathrm{Ai}(\lambda+\xi_i+a^{2}),\quad i=1,2.
\end{aligned},
\end{equation}
where $|\varepsilon_{0}(a)|\leq 1/4a$, and, for $i=1,2$, $|\varepsilon_{i}(a)|\leq \max_{\lambda\geq 0} \mathrm{Ai}(\lambda+\xi_i+a^{2})/a$ and $c_{i}=\int_{0}^{\infty}\dx \lambda \mathrm{Ai}(\xi_{3-i}+a^{2}+\lambda)$.
This implies that when $a$ increases, the ratio between the bounded terms and the large term becomes smaller than $10^{-16}$ machine precision and no reliable numerical evaluation is possible. In our case, already for $a\geq 4$, $\widetilde{K}_a$ (much less $G_a$) can no longer be computed in Matlab. For instance, Matlab computes $G_3(s)={\rm NaN}$ for all tested $s$, $G_{2.5} (-1)= 0.0838$ with an error $0.0044$, whereas $G_{1.8}(-2)=1.4879\cdot 10^{-4}$, with an error $5.6831\cdot 10^{-9}$. We present the numerical computations until $a=1.8$, since for higher values the error term in the Kurtosis becomes visible.

Generally, the computational error of $G_a(s)$ decreases as $s$ increases since then $\widetilde{K}_{a}(\xi_1,\xi_2)$ needs to be computed only for small entries and the evaluation of $G_{a}$ is easier (namely, the matrix whose determinant approximates $G_a$ gets closer to the Identity matrix as $s$ increases, see (4.3) in \cite{Born10}). The statistics of $G_a$ were computed using the \texttt{chebfun} package (see~\cite{BF04}), in which $G_a$ is represented by its polynomial interpolant in Chebyshev points, for our choice in $n=4096$ points.

In Figure~\ref{FigDistr} we plot $F_1(2s)=G_0(s)$, $G_a(s)$ for $a\in \{0.1,0.2,\ldots,1.8\}$ and the conjectured $a\to\infty$ limit, namely $(F_1(2s))^2$.
A property which is apparent from Figure~\ref{FigDistr} is that $G_{a}$ monotonically decreases towards $F_{1}(2s)^{2}$ as $a$ grows. Indeed, for all $a,a'\in\{0,\ldots,1.8\}$, and $s\in\{-2,-1.9,\ldots,2\}$ we have
\begin{equation} \label{evid1}
G_{a}(s) > F_{1}(2s)^{2}, \quad G_{a'} (s) < G_{a}(s)\quad {\rm if}\quad a < a'.
\end{equation}
An analytic proof of this property does not seem to be trivial and is not available so far.

To further quantify the difference between $F_{1.8}$ and $F_{1}(2\cdot)^{2}$ we computed
\begin{equation} \label{aprx}
D(a):=\max_{s=-2,-1.9,\ldots,2} |F_{1}(2s)^{2}-G_{1.8}(s)|.
\end{equation}
 \begin{figure}
 \begin{center}
 \includegraphics[height=5cm]{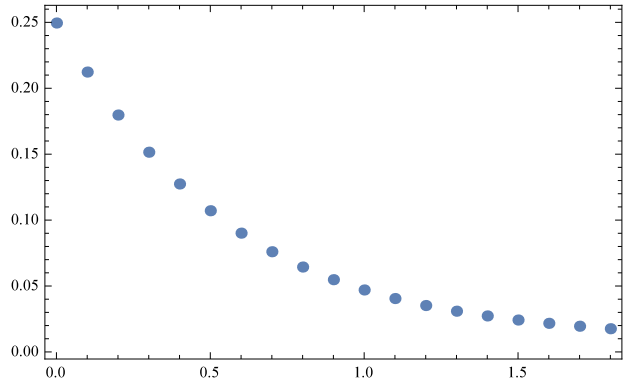}
 \caption{Plot of the function $a\mapsto D(a)$ defined by \eqref{aprx}. On this limited interval width for $a$ the convergence of the difference of the distribution functions close to exponentially fast.}\label{FigMaxDiff}
 \end{center}
 \end{figure}

\eqref{evid1} and \eqref{aprx} are compatible with the conjecture \eqref{eqConjEquiv}, but to have a further more reliable verification we study numerically the basis statistics too. The reason is that the distribution functions might be optically close but still be different. For example,  the plots of the GUE and GOE Tracy-Widom distribution functions scaled to have both average $0$ and variance~$1$, are almost indistinguishable. However, by looking at their skewness and kurtosis one can clearly differentiate between them.

In Figure~\ref{moments} we plot the basic statistics of $G_a$ and compare them with those of $F_{1}(2\, \cdot)^{2}$. The approximation is fastest for the expectation, and slowest for the kurtosis (though the observation window for $a$ is too small to  quantify the different rates of convergence).
\begin{figure}
\begin{center}
 \includegraphics[height=5.3cm]{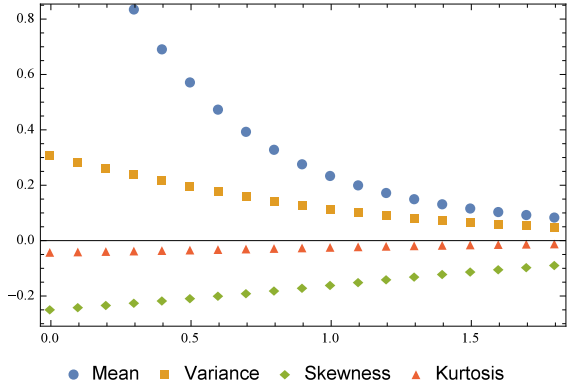}
 \caption{Relative difference between the basic statistics of $G_a$ and of the conjectured limiting distribution, $F_1(2\cdot)^2$.}
 \label{moments}
\end{center}
\end{figure}
Finally, let us resume in  Table~\ref{tab} the basic statistics of $G_a$ in comparison to $F_{1}(2\cdot)^{2}$.
\begin{table}
\begin{center}
 \begin{tabular}{ | l | l | l | l | l | l |}
\hline
& Expectation & Variance & Skewness & Kurtosis \\
\hline
$G_0$ & $-0.6033; 144\%$ & $0.4019; 31\%$ & $0.2935; -25\%$ & $3.1652; -4.3\%$ \\
\hline
$G_{0.3}$ & $-0.4524; 83\%$ & $0.3816; 24\%$ & $0.3028; -23\%$ & $3.1811; -3.9\%$ \\
\hline
$G_{0.6}$ & $-0.3632; 47\%$ & $0.3624; 17\%$ & $0.3127; -20\%$ & $3.1988; -3.3\%$ \\
\hline
$G_{0.9}$ & $-0.3145; 27\%$ & $0.3466; 13\%$ & $0.3240; -17\%$ & $3.2184; -2.7\%$ \\
\hline
$G_{1.2}$ & $-0.2889; 17\%$ & $0.3353; 8.9\%$ & $0.3359; -14\%$ & $3.2377; -2.1\%$ \\
\hline
$G_{1.5}$ & $-0.2751; 11\%$ & $0.3277; 6.4\%$ & $0.3469; -11\%$ & $3.2540; -1.6\%$ \\
\hline
$G_{1.8}$ & $-0.2670; 8.2\%$ & $0.3226; 4.7\%$ & $0.3561; -9.1\%$ & $3.2658; -1.3\%$ \\
\hline
$\boldsymbol{F_{1}(2\cdot)^{2}}$& $-0.2468$ & $0.3080$ & $0.3917$ & $3.3086$ \\
\hline
\end{tabular}
\caption{Data of the basic statistics and their relative difference to the conjectured limit distribution $F_1(2\,\cdot)^2$ for a few values of $a$.}\label{tab}
\end{center}
\end{table}

\section{Asymptotic analysis - Proof of Theorem~\ref{ThmMain}}\label{SectAsympt}
In Section~\ref{SectFiniteKernel} we derive the finite time kernel, whose Fredholm determinant gives us the joint distributions of TASEP particle positions, see Proposition~\ref{fdetform}. For the derivation we first need to consider the case of a finite number $M$ of $\alpha$-particles and then take the $M\to\infty$ limit. In Section~\ref{SectScalingLimit} we then perform the asymptotic analysis and complete the proof of Theorem~\ref{ThmMain}.

\subsection{Finite time formula}\label{SectFiniteKernel}
Taking the limit of the situation with finitely many slow particles we obtain the following result.
\begin{prop}\label{fdetform}
Consider Two-Speed TASEP as defined in \eqref{twospeedTASEP}. Then the joint distribution of the positions of $m$ normal particles with labels \mbox{$0< n_1 < n_2 < \ldots < n_m$} at time $t$ is given by
\begin{equation}
\Pb\left( \bigcap_{k=1}^{m} \{x_{n_k} (t) > s_k\} \right)=\det(1-\chi_s K \chi_s)_{\ell^{2}(\{n_1,\ldots,n_m\}\times\mathbb{Z})},
\end{equation}
with $K=-\phi+ K^{1}+K^{2}$, where $\chi_{s}(n_k,x)=\Id_{(-\infty,s_k]}(x)$ and\footnote{For a set $S$, the notation $\Gamma_{S}$ means a simple path anticlockwise oriented enclosing only poles of the integrand belonging to the set $S$.}
\begin{equation}
\phi(n_1,x_1;n_2,x_2)=\frac{1}{2\pi\I}\oint_{\Gamma_{0}}\frac{\dx w}{w}\frac{(w-1)^{n_1-n_2}}{w^{x_1-x_2+n_1-n_2}}\Id_{\{n_1 < n_2\}},
\end{equation}
\begin{equation}
\begin{aligned}
K^{1}(n_1,x_{1};n_2,x_{2})=&\frac{1}{(2\pi\I)^{2}}\oint_{\Gamma_0}\dx v\oint_{\Gamma_{0,-v}}\frac{\dx w}{w}\frac{e^{tw}(w-1)^{n_1 }}{w^{x_1 + n_1}}\\&
\times\frac{(1+v)^{x_2 + n_2 }}{e^{t(v+1)}v^{n_2}}\frac{1+2v}{(w+v)(w-v-1)},
\end{aligned}
\end{equation}
and
\begin{equation}
\begin{aligned}
K^{2}(n_1,x_{1}; n_2, x_{2})=&\frac{-1}{(2\pi\I)^2}\oint_{\Gamma_{0}}\dx w\oint_{\Gamma_{0,\alpha-1-w}}\dx v \frac{e^{tw}(w-1)^{n_1}}{w^{x_{1}+n_1+1}}
 \frac{(1+v)^{x_{2}+n_2}}{e^{t(v+1)}v^{n_2}}\\&\times\frac{1+2v}{(v+w+1-\alpha)(w-v-\alpha)}.
\end{aligned}
\end{equation}
\end{prop}
The system with finitely many slow particles has been already partially studied in~\cite{BFS09}. There, it was shown that the distribution function of particles positions is given by a Fredholm determinant and the kernel was given. For a fixed $M\in \N$, consider TASEP with initial conditions and jump rates given by
\begin{align} \label{f2spTASEP}
x_{n}^{M}(0)=2(M-n),\quad n\in \mathbb{N},\qquad v_{n}^{M}=
\left\{\begin{array}{ll}
1, & \textrm{for }n>M, \\
\alpha, & \textrm{for }1\leq n\leq M.
\end{array}\right.
\end{align}
To distinguish this system with the one we want to study, i.e., $M=\infty$, we will index all quantities by a $M$. Proposition~6 of~\cite{BFS09} tells us that
\begin{equation}\label{ftime}
\Pb\left( \bigcap_{k=1}^{m} \{x_{n_k +M}^{M} (t) > s_k\} \right)=\det(1-\chi_s K^{M} \chi_s)_{\ell^{2}(\{n_1,\ldots,n_m\}\times\mathbb{Z})},
\end{equation}
where the kernel kernel $K^{M}$ has the decomposition
\begin{align}\label{3some}
K^{M}=-\phi+K^{1}+K^{2,M}.
\end{align}
Here $\phi$ and $K^{1}$ are as in Proposition~\ref{fdetform}, while $K^{2,M}$ is given by
\begin{equation}\label{3someB}
\begin{aligned}
K^{2,M} (n_1,x_1;n_2, x_2)=&\frac{1}{(2\pi\I)^{3}}\oint_{\Gamma_{\alpha-1}}\dx v \oint_{\Gamma_{0,v}}\dx z\oint_{\Gamma_{0,\alpha-1-v}} \frac{\dx w}{w}\frac{e^{tw}(w-1)^{n_1}}{w^{x_1+n_1}}\\&\times \frac{(1+z)^{x_2+n_2}}{e^{tz}z^{n_2}}\left(\frac{w(w-\alpha)}{(v+1)(v+1-\alpha)}\right)^{M}
\\& \times \frac{(1+2z)(2v+2-\alpha)}{(z-v)(z+v+1)(w-1-v)(w+1-\alpha+v)}.
\end{aligned}
\end{equation}

\begin{proof}[Proof of Proposition~\ref{fdetform}]
First we note that
\begin{equation}
\lim_{M\to \infty} \Pb\left(\bigcap_{k=1}^{m}\{x_{n_k+M}^{M}(t) > s\}\right)=\Pb\left(\bigcap_{k=1}^{m} \{x_{n_k}(t) > s\}\right).
\end{equation}
This follows since $x_{n+M}^{M}(0)=x_{n}(0)$ for all $n\geq -M$ and by the fact that in TASEP the positions of the normal particles up to a fixed time $t$ depend only on finitely many other particles on the right with probability one, as is seen from a graphical construction of it.
So it remains to show that the convergence in \eqref{kernconv} holds also on the level of fredholm determinants.

First of all, as shown already in Corollary~8 of~\cite{BFS09}), it holds
\begin{equation}\label{kernconv}
\lim_{M\to \infty} K^{M}(n_1+M,x_1;n_2+M,x_2)= K(n_1,x_1;n_2,x_2)
\end{equation}
pointwise. The reason being that for any $M>(n_1+x_1+1)$ the pole at $w=0$ in \eqref{3someB} vanishes, in the limit of large $M$ we can integrate out explicitly the simple pole at $w=\alpha-1-v$ of the kernel $K^{2,M}$ and it results in the kernel $K^2$ given in Proposition~\ref{fdetform}.

To show the convergence of Fredholm determinants we use their series expansion expression, namely
\begin{multline}\label{eq79}
\det(\Id-\chi_s K^{M} \chi_s)_{\ell^2(\{n_1,\ldots,n_m\}\times\Z)}\\
 = \sum_{n\geq 0}\frac{(-1)^n}{n!} \sum_{i_1,\ldots,i_n =1}^{m}\sum_{x_1\leq s_1}\ldots \sum_{x_n\leq s_n} \det[K^{M}(n_{i_k},x_k;n_{i_l},x_l)]_{1\leq k,l\leq n}.
\end{multline}
It is easy to see that $K^1(n_1,x_1;n_2,x_2)=0$ for $x_1<-2 n_1$ since the pole at $w=0$ vanishes after computing the residue at $w=-v$ the pole at $v=0$ vanishes. Similarly, $K^{2,M}(n_1,x_1;n_2,x_2)=0$ for $x_1<-n_1$ since the pole at $w=0$ vanishes. Further, in the term $\phi(n_1,x_1;n_2,x_2)$, if $x_2$ is bounded from below, then for $x_1$ small enough this term is also zero. This implies that the $n\times n$ determinant in \eqref{eq79} is strictly equal to zero if $x_i<-2 n_m$. The physical reason for this is that if we consider the system with particle numbers bounded from above by $n_m$, then by TASEP dynamics particles can be present only in the region on the right of $x_{n_m}(0)=-2n_m$. Consequently, the sums are finite and the by Hadamard bound $|\det[K^{M}(n_{i_k},x_k;n_{i_l},x_l)]_{1\leq k,l\leq n}|\leq C^n n^{n/2}$ for some finite constant $C$. Thus by dominated convergence we can take the limit $M\to\infty$ inside the sum and the proof is completed.
\end{proof}

\subsection{Scaling limit and asymptotics}\label{SectScalingLimit}
With the finite time formula of Proposition~\ref{fdetform} at hand, we can now proceed to prove the main result.
\begin{proof}[Proof of Theorem~\ref{ThmMain}]
The proof is identical to the one of Theorem 2.5 in~\cite{BFP06}, given that the we have convergence of the (properly rescaled) kernel in a bounded set (Proposition~\ref{PropPointwiseConv}), and good enough bounds to control the convergence of the Fredholm determinant by the use of dominated convergence. The strategy, nowadays standard, was first used by Tracy, Widom and Gravner in~\cite{GTW00}. The bounds are contained in Propositions~\ref{PropBoundPhi},\,~\ref{PropModDev}, and~\ref{PropLargeDev} below.
\end{proof}

From now on, the $u_i$ are some fixed real values. We first prove convergence to the limit kernel $K_a$ and then provide integrable bounds.
We consider the scaling
\begin{equation}\label{shockscaling}
\begin{aligned}
n_i (u,t)&=t/4 +(u_{i}+a)\left(t/2\right)^{2/3},\\
x_{i}(u,t)&=-2(u_i+a)\left(t/2\right)^{2/3}- \xi_{i}\left(t/2\right)^{1/3}.
\end{aligned}
\end{equation}
The $\xi_i$ measure the fluctuations in the $(t/2)^{1/3}$ scale with respect to the macroscopic approximation given in \eqref{eq2.6}. Accordingly, we define the rescaled kernel
\begin{equation}
K^{{\rm resc}}(u_1,\xi_1;u_2,\xi_2)=2^{x_2-x_1}(-1)^{n_1-n_2}\left(t/2\right)^{1/3} K(n_1,x_1;n_2,x_2)
\end{equation}
and similarly for each component of the kernel.

\begin{prop}[Convergence on bounded sets]\label{PropPointwiseConv}
For any fixed $L>0$, we have
\begin{equation}
\begin{aligned}
\lim_{t\to \infty}K^{\rm resc}(u_1,\xi_1;u_2,\xi_2)=K_{a}(u_1,\xi_1;u_2,\xi_2),
\end{aligned}
\end{equation}
uniformly for $\xi_1,\xi_2$ in $[-L,L]$. Here $K_a$ is the kernel from Theorem~\ref{ThmMain}.
\end{prop}
\begin{proof}
We start with $\phi$. The residue at $0$ can be easily computed expanding $(w-1)^{n_1-n_2}$ with the binomial formula and one readily obtains that $\phi(n_1,x_1;n_2,x_2)=(-1)^{n_1-n_2}{x_1-x_2-1\choose n_2-n_1-1}$. It is then an easy computation to show that (see e.g. Proposition~7 of~\cite{BFS07})
\begin{equation}
\phi^{\rm resc}(u_1,\xi_1;u_2,\xi_2)\to\frac{\Id_{\{u_2>u_1\}}}{\sqrt{4\pi(u_2-u_1)}}\exp\left(-\frac{(\xi_2-\xi_1)^{2}}{4(u_2-u_1)}\right).
\end{equation}

Next we consider $K^1$. We make the change of variables $w\to w+1$, rename $u=w$, and set $\tau_i = u_1+a$ and $\tilde{s}_i=\xi_i$. Then,
$K^{1,{\rm resc}}$ equals the kernel $\widehat{K}_{t}^{{\rm resc}}$ in (3.7) of~\cite{BFS07}. The convergence of $\widehat{K}_{t}^{{\rm resc}}$ to the $\mathcal{A}_{2\to 1}$ transition kernel is proven in Proposition 4 of~\cite{BFS07}, giving the first double integral of $K_a$, i.e., \eqref{vanishing}+\eqref{div1} in the integral representation of Appendix~\ref{SectAppendix}.

Finally consider $K^2$. We have
\begin{multline}\label{expart}
K^{2,{\rm resc}}(u_1,\xi_1;u_2,\xi_2)\\
=-\frac{(t/2)^{1/3}}{(2\pi\I)^{2}}\oint_{\Gamma_{-1}}\dx u\oint_{\Gamma_{0,\alpha-2-u}}\dx v
\frac{1+2v}{(v+u+2-\alpha)(u+1-v-\alpha)}\\
\times \frac{e^{t f_{0}(v) + (t/2)^{2/3} (a+u_2) f_{1}(v)+(t/2)^{1/3}\xi_2 f_{2}(v)}}{e^{t f_{0}(u) + (t/2)^{2/3} (a+u_1) f_{1}(u)+(t/2)^{1/3}\xi_1 f_{2}(u)+f_3 (u)}}
\end{multline}
with
\begin{equation}
\begin{aligned}
f_{0}(v)&=-v+\frac{1}{4}\ln((1+v)/v),\\
f_{1}(v)&=-\ln(-4v(1+v)),\\
f_{2}(v)&=-\ln(2(1+v)),\\
f_{3}(v)&=\ln(1+v).
\end{aligned}
\end{equation}
The poles and order of integration are different, but the exponential part \eqref{expart} equals again the exponential part of $\widehat{K}_{t}^{{\rm resc}}$ in (3.7) of~\cite{BFS07}, so let us focus on the differences.
The critical point of $f_0$ is $-1/2$, and in Proposition~4 of~\cite{BFS07}
$\mathbb{C}$ is divided in four regions $D_i$ depending on the sign of $\Re(v)+1/2$ and of $\mathrm{Re}(f_{0}(v)-f_{0}(-1/2)))$, see Figure~\ref{FigZeros}.
\begin{figure}
\begin{center}
\psfrag{-1}[c]{$-1$}
\psfrag{-0.5}[c]{$-0.5$}
\psfrag{ 0}[c]{$0$}
\psfrag{-0.2}[c]{$-0.2$}
\psfrag{ 0.2}[c]{$0.2$}
\psfrag{D1}[c]{$D_1$}
\psfrag{D2}[c]{$D_2$}
\psfrag{D3}[c]{$D_3$}
\psfrag{D4}[c]{$D_4$}
\psfrag{x}[c]{$x$}
\psfrag{y}[c]{$y$}
\includegraphics[height=6cm]{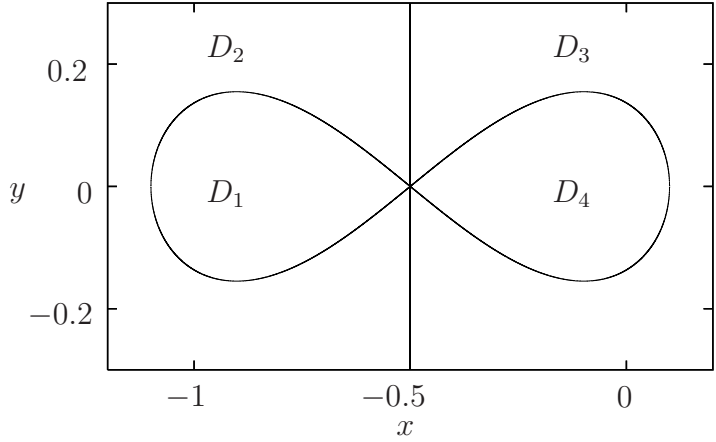}
\caption{The signum of $\Re(f_0(x+\I y)-f_0(-1/2))$ is positive in $D_2$ and $D_4$ and negative in $D_1$ and $D_3$.}
\label{FigZeros}
\end{center}
\end{figure}
For $\Gamma_{0,\alpha-2-u}$ we may choose any simple anticlockwise oriented closed path passing through $-1/2$ and staying in $D_3$. $\Gamma_{-1}$ is restricted to stay in $D_2$ except for a local modification in a \mbox{$t^{-1/3}-$neighborhood} of the critical point in order to satisfy \mbox{$\alpha-2-\Gamma_{-1}\subset \Gamma_{0,\alpha-2-u}$}. More precisely, $\Gamma_{-1}$ passes through $-1/2-\kappa/t^{1/3}$ for some $\kappa> 2^{4/3}a$, see Figure~\ref{FigPaths1}.
\begin{figure}
\begin{center}
 \psfrag{pi6}[l]{$\pi/6$}
 \psfrag{phi}[l]{$\varphi$}
 \psfrag{delta}[c]{$\kappa/t^{1/3}$}
 \psfrag{G1}[l]{$\Gamma_{-1}$}
 \psfrag{G11}[l]{$\Gamma_{0,\alpha-2-u}$}
 \includegraphics[height=5cm]{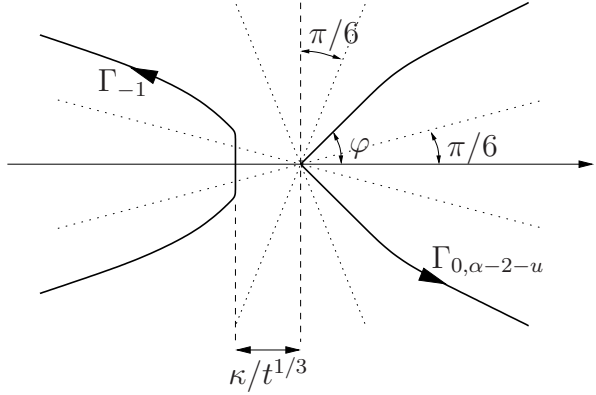}
 \caption{The contours $\Gamma_{-1}$ and $\Gamma_{0,\alpha-2-u}$ used for the pointwise convergence. The point in the middle is $(-1/2,0)$. The vertical piece in $\Gamma_{-1}$ is of length of order $t^{-1/3}$.}
 \label{FigPaths1}
\end{center}
\end{figure}
We will take $\Gamma_{0,\alpha-2-u}$ to arrive in $-1/2$ with an angle $\varphi\in(\pi/6,\pi/3)$.

Define for $\delta>0$ the segments $\Gamma_{0,\alpha-2-u}^{\delta}=\{v\in\Gamma_{0,\alpha-2-u}:|1/2+v|<\delta\}$ and $\Gamma_{-1}^{\delta}=\{u\in\Gamma_{-1}:|1/2+u|<\delta\}$.
Denote by $\Sigma$ the part of the contours where $v\notin\Gamma_{0,\alpha-2-u}^{\delta}$ and/or $u\notin\Gamma_{-1}^{\delta}$.
 Then the integral is on
\begin{equation}\Sigma+(\Gamma_{0,\alpha-2-u}^{\delta}\cup\Gamma_{-1}^{\delta}) =\Gamma_{0,\alpha-2-u}\cup\Gamma_{-1}
\end{equation}
On $\Sigma$ there exists a $c_0 >0$ that $\Re(f_{0}(v)-f_{0}(-1/2))\leq -c_0$ and/or
$\Re(-f_{0}(u)+f_{0}(-1/2)\leq -c_0$. Further $\exp(t(f_{0}(-1/2-\kappa/t^{1/3})-f_{0}(-1/2)))=\Or(1)$. Hence the contribution coming from $\Gamma_{0,\alpha-2-u}\setminus\Gamma_{\alpha-2-u}^{\delta}$ and $\Gamma_{-1}\setminus\Gamma_{-1}^{\delta}$ is bounded by
$e^{-c_{0} t+\Or(t^{2/3})}$. Furthermore, on $\Sigma$, $\left|\frac{1+2v}{(v+u+2-\alpha)(u+1-v-\alpha)}\right|\leq C (\delta)$ with $C(\delta)$ depending only on $\delta$. Hence we may bound the overall contribution of $\Sigma$ by
\begin{equation}\label{Sigma}
\left|\int_{\Sigma}\cdots\right|\leq c_1 t^{1/3} C(\delta) e^{-t c_{0}/4}
\end{equation}
for some finite constant $c_1$. As we will show below, the contribution coming from $\Gamma_{0,\alpha-2-u}^{\delta}$ and $\Gamma_{-1}^{\delta}$ is of order one, therefore the contribution of the integrals over $\Sigma$ is negligible in the $t\to\infty$ limit.

Next consider the contribution from the integral over $\Gamma_{0,\alpha-2-u}^{\delta}\cup\Gamma_{-1}^{\delta}$. Consider the change of variables
\begin{equation}
u=-1/2+(U-2a)/(4t)^{1/3}, \quad v=-1/2+(V-2a)/(4t)^{1/3}
\end{equation}
and denote $F_i(v)=e^{t f_{0}(v) + (t/2)^{2/3} (a+u_i) f_{1}(v)+(t/2)^{1/3}\xi_i f_{2}(v)+(2-i)f_3 (v)}$. Then by Taylor expansion we obtain
\begin{equation}\label{TaylorB}
\begin{aligned}
\frac{F_2(v)}{F_1(u)}=&2\frac{e^{V^{3}/3+(u_2-a) V^{2}-(\xi_2+4 a u_2) V+4 u_2 a^2}}{e^{U^{3}/3+(u_1-a) U^{2}-(\xi_1+4 a u_1) U+4 u_1 a^2}}
\frac{e^{\Or(V^{2}/t^{1/3})+\Or(V^{3}/t^{1/3})+\Or(V^{4}/t^{1/3})}}{e^{\Or (U/t^{1/3})+\Or(U^{2}/t^{1/3})+\Or(U^{3}/t^{1/3})+\Or(U^{4}/t^{1/3})}}
\end{aligned}
\end{equation}
The control of the error term in \eqref{TaylorB} is (almost) identical to the one given in Proposition~4 of~\cite{BFS07}, we therefore omit it. The error term is of order $\Or(t^{-1/3})$. For the remaining part, denote $\gamma_{+}^{\delta}=(4t)^{1/3}(\Gamma_{0,\alpha-2-u}^{\delta}+1/2)+2a$ and
$\gamma_{-}^{\delta}=(4t)^{1/3}(\Gamma_{-1}^{\delta}+1/2)+2a$ . Any extension of finite length $\gamma_{+}^{\delta},\gamma_{-}^{\delta}$ gives an error of order \eqref{Sigma}. For $|v|$ large, $\Re(f_{0}(v)-f_{0}(-1/2))$ (resp.\ $\Re(-f_{0}(v)+f_{0}(-1/2))$) decays linearly along $\gamma_{+}^{\delta}$ (resp.\ $\gamma_{-}^{\delta}$). Therefore also extending the curves to infinity creates an error of order $e^{-c t}$ for some $c>0$. We denote the resulting curves by $\gamma_{+}, \gamma_{-}$ and we are thus left with
\begin{equation}
\begin{aligned}
\frac{-1}{(2\pi\I)^{2}}\int_{\gamma_{-}} \dx U \int_{\gamma_{+}} \dx V
&\frac{e^{V^{3}/3+(u_2-a) V^{2}-(\xi_2+4 a u_2) V+4 u_2 a^2}}{e^{U^{3}/3+(u_1-a) U^{2}-(\xi_1+4 a u_1) U+4 u_1 a^2}}
\frac{2(V-2a)}{(V+U)(U-V+4a)}.
\end{aligned}
\end{equation}
The integration paths can be deformed as in Definition~\ref{defLimitProcess} without errors (the minus factors comes from the change of orientation of one of the paths).

\end{proof}

For $\phi$, an integrable bound was already obtained in~\cite{BFS07} (with $\phi$ as binomial coefficient, see the beginning of the proof of Proposition~\ref{PropPointwiseConv}).
\begin{prop}[Proposition 8 in~\cite{BFS07}]\label{PropBoundPhi}
For any $\xi_1,\xi_2$ in $\mathbb{R}$ and $u_2-u_1>0$ fixed, there exist a finite constants $C$ and $t_0$, such that for all $t\geq t_0$,
\begin{equation}
0\leq \phi^{\rm resc}(u_1,\xi_1;u_2,\xi_2)\leq C\, e^{-|\xi_2-\xi_1|}.
\end{equation}
\end{prop}

\begin{prop}[Moderate deviations for $K^1,K^2$]\label{PropModDev} For any $L$ large enough, there are $\varepsilon_{0}$, $t_0$ such that for all $0<\varepsilon\leq \varepsilon_{0}$, $t\geq t_0$, there exists a finite constant $C$ such that
\begin{equation}
\left|K^{1,{\rm resc}}(u_1,\xi_1;u_2,\xi_2)+K^{2,{\rm resc}}(u_1,\xi_1;u_2,\xi_2)\right| \leq C e^{-(\xi_1+\xi_2)/2}
\end{equation}
for all $\xi_1, \xi_2 \in [-L,\varepsilon t^{2/3}]\setminus[-L,L]$.
\end{prop}

\begin{proof}
For $K^{1,{\rm resc}}$ the statement is Proposition~5 in~\cite{BFS07}. For $K^{2,{\rm resc}}$, we follow a similar strategy, but let us give the details. Define $\sigma_i = \xi_i t^{-2/3} 2^{-1/3} \in (0, \varepsilon]$ and denote the integrand by
$G_{\sigma_1,\sigma_2}(u,v):=\frac{F_{2}(v)}{F_{1}(u)}\frac{(t/2)^{1/3}(1+2v)}{(v+u+2-\alpha)(u+1-v-\alpha)}$. Let $\mathcal{I}$ be an interval on which $\Gamma_{-1},\Gamma_{0,\alpha-2-u}$ are parametrized. The analysis of Proposition~\ref{PropPointwiseConv} shows that for a constant $C$
\begin{multline}\label{properbound}
|K^{2,{\rm resc}}(n_1,0;n_2,0)|\\
\leq \int_{\mathcal{I}^{2}}\dx s \dx r |\Gamma_{-1}'(s)\Gamma_{0,\alpha-2-u}'(r)G_{0,0}(\Gamma_{-1}(s),\Gamma_{0,\alpha-2-u}(r))|\leq C.
\end{multline}

If $\sigma_i >0$, we have an additional factor
\begin{equation}\label{exfac}
\exp(-t\sigma_2\ln(2+2v))\exp(t\sigma_1\ln(2+2u))
\end{equation}
in the integrand of \eqref{properbound}. As we shall show in (b), (c), (e), (f) below, if we are not close to $-1/2$, then $|\eqref{exfac}|\leq e^{-(\xi_1+\xi_2)/2}$ and thus get the bound $C e^{-(\xi_1+\xi_2)/2}$. Close to $-1/2$, we do a modification of one of the contours, depending on whether
 $\sigma_1 \leq \sigma_2$ or $\sigma_1 \geq \sigma_2$, and then get the needed decay for \eqref{exfac}.

In the $\sigma_1\geq\sigma_2$ case, we modify $\Gamma_{-1}$ near the critical point $-1/2$ and show that in the unmodified region the decay is the same as in the case $\sigma_1=\sigma_2=0$ case times an integrable factor. We then deal with the modified region and provide the needed decay there too. If $\sigma_1\leq\sigma_2$, we integrate out the residue at $v=\alpha-2-u,$ and show the needed decay for it by modifying $\Gamma_{-1}$. In the remaining integral we may then deform the contour $\Gamma_{0,\alpha-2-u}$ to get the desired decay.

\textbf{Case} $\sigma_1 \geq \sigma_2$. The paths $\Gamma_{0,\alpha-2-u}$ and $\Gamma_{-1}$ are as in Figure~\ref{FigPaths1} except that the distance of the vertical piece of $\Gamma_{-1}$ with respect to $-1/2$ is $\sqrt{\sigma_1}/2+\kappa/t^{1/3}$ instead of $\kappa/t^{1/3}$.
Near $-1/2$ we modify $\Gamma_{-1}$ by a vertical part $\Gamma_{{\rm vert}}$ that passes through $-1/2-\mu$ with $\mu\ll 1$ (see \eqref{vert}) and which is symmetric w.r.t.  the real line. As in Proposition~\ref{PropPointwiseConv}
let $\varphi\in (\pi/6,\pi/3)$ be the angle with which $\Gamma_{0,\alpha-2-u}$ leaves $-1/2$ and let $\kappa > 2^{4/3}a$. The region $D_1$ in Figure~\ref{FigZeros} leaves $-1/2$ with angle $\pm 5\pi/6$. Consequently,
for $\Gamma_{{\rm vert}}$ to end  outside $D_1$ and satisfy $\alpha-2-\Gamma_{{\rm vert}}\subset \Gamma_{\alpha-2-u}$ we can choose (for $t$ large enough) its length as $\mu b$ for some $b\in (\tan(\pi/6),\tan(\varphi))$. Hence we define
\begin{equation}\label{vert}
\Gamma_{{\rm vert}}=\{-1/2-(\sqrt{\sigma_1}/2+\kappa/t^{1/3})(1+\I \rho), \rho \in [-b,b]\}.
\end{equation}

(a) The choice of contours is such that
\begin{equation}
\begin{aligned}
{\rm dist}(\Gamma_{0,\alpha-2-u},\Gamma_{-1}+1-\alpha)&\geq c_3 \sqrt{\sigma_1 },\\
{\rm dist}(-\Gamma_{0,\alpha-2-u},\Gamma_{-1}+2-\alpha)&\geq c_3 \sqrt{\sigma_1 }.
\end{aligned}
\end{equation}
for some constant $c_{3}=c_{3}(b,\varphi)>0$.
This is at least the same order as
 for the contours in Proposition ~\ref{PropPointwiseConv} where we had
\begin{equation}
\begin{aligned}
{\rm dist}(\Gamma_{0,\alpha-2-u},\Gamma_{-1}+1-\alpha)&\leq (\kappa-2^{4/3}a)/t^{1/3},\\
{\rm dist}(-\Gamma_{0,\alpha-2-u},\Gamma_{-1}+2-\alpha)&\leq (\kappa-2^{4/3}a)/t^{1/3}.
\end{aligned}
\end{equation}
Hence (as in the \mbox{$\sigma_1=\sigma_2=0$} case) the
$\big|\frac{1+2v}{(v+u+2-\alpha)(u+1-v-\alpha)}\big|$ term does not create problems .

(b) The contour $\Gamma_{0,\alpha-2-u}$ can be chosen such that $|1+v|$ reaches its minimum at $v=-1/2$ so we can simply bound
\begin{equation}
|e^{-t\sigma_2 \ln(2(1+v))}|\leq 1.
\end{equation}

(c) Let $u\in\Gamma_{-1}\setminus\Gamma_{{\rm vert}}$. In the following, we set \mbox{$\widehat\sigma_1 := (\sqrt{\sigma_{1}}+2\kappa/t^{1/3})^{2},$} which is just a shift in the variable $\sqrt{\xi_1}$. $\Gamma_{-1}$ can be chosen such that on
$\Gamma_{-1}\setminus\Gamma_{{\rm vert}}$
the maximum of $|1+u|$ is reached at $\rho=\pm b$ . For $\varepsilon$ small enough
\begin{equation}
(2|1+u|)^{2}=1-2\sqrt{\widehat\sigma_1}+(b^{2}+1)\widehat\sigma_1\leq 1-\sqrt{\sigma_1}.
\end{equation}
Therefore it holds
\begin{equation}
|e^{ t \sigma_1 \ln(2(1+u)) }| \leq e^{t\sigma_1 \ln(1-\sqrt{\sigma_1})/2}\leq e^{-\xi_{1}^{3/2}/2^{3/2}+\Or(t \sigma_{1}^{2})}\leq e^{-\xi_{1}^{3/2}/4}
\end{equation}
for $\e$ small enough.

(d) For the integral on $\Gamma_{\mathrm{vert}}$, it is an integral on $[-b,b]$ in the variable $\rho$. Since $\Gamma_{{\rm vert}}'(\rho)=(\sqrt{\sigma_1}/2+\kappa/t^{1/3})\I$, this term multiplied by the $t^{1/3}$ prefactor gives a term $\Or(\xi_1^{1/2})$. So it suffices to have a bound on the integrand that controls it. On $\Gamma_{{\rm vert}}$ we use Taylor expansion around $-1/2$ (from which $\Gamma_{{\rm vert}}$ is at most $\Or(\sqrt{\varepsilon})$ far away). The $u-$dependant part of the exponential term becomes
\begin{equation}\label{est}
e^{-t f_{0}(-1/2)+t\widehat\sigma_{1}^{3/2}(1+\I\rho)^{3}/6-u_1 (t/2)^{2/3}\widehat\sigma_1(1+\I\rho)^{2}}
e^{-t\sigma_1 \sqrt{\widehat\sigma_1}(1+\I\rho)+\Or(t \widehat\sigma_{1}^{2})}.
\end{equation}
Now we take real parts in the exponent. We see that
for $L$ large and $\varepsilon$ small enough we have $\xi_{1}^{3/2}\gg t\widehat\sigma_{1}^{2}$ and $\xi_{1}^{3/2}\gg \,(2^{-1/6}\sqrt{\xi_1}+2\kappa)^{2}$.
We get the upper bound
\begin{equation}
\begin{aligned}
|\eqref{est}|&\leq e^{-tf_{0}(-1/2)+\xi_{1}^{3/2}(-5/6-\rho^{2}/2)}e^{-u_{1} \frac{(2^{-1/6}\sqrt{\xi_1}+2\kappa)^{2}}{2^{2/3}}(1-\rho^{2})}e^{-\xi_1 \kappa 2^{1/3}}e^{\Or(t \widehat\sigma_{1}^{2})}\\&\leq e^{-tf_0 (-1/2)-\xi_{1}^{3/2}/4}.
\end{aligned}
\end{equation}
The $e^{-tf_0 (-1/2)}$ cancels exactly with the contribution coming from the integrand in the $v$ variable. Finally note that for $L$ large enough
\begin{equation}\label{fini}
e^{-\xi_{1}^{3/2}/4}\leq e^{-\xi_{1}\sqrt{L}/4}\leq e^{-(\xi_{1}+\xi_2)/2}.
\end{equation}

\textbf{Case} $\sigma_1 \leq \sigma_2$. Here we integrate out the residue $w=\alpha-2-u$ and obtain
\begin{multline}\label{nores}
K^{2,{\rm resc}}(n_1,\xi_1;n_2,\xi_2)\\
=\mathcal{I}_1 - (t/2)^{1/3}\frac{1}{(2\pi\I)^{2}}\oint_{\Gamma_{-1}}\dx u\oint_{\Gamma_{0}}\dx v
\frac{1+2v}{(v+u+2-\alpha)(u+1-v-\alpha)}\\
 \times \frac{e^{t f_{0}(v) + (t/2)^{2/3} (a+u_2) f_{1}(v)+(t/2)^{1/3}\xi_2 f_{2}(v)}}{e^{t f_{0}(u) + (t/2)^{2/3} (a+u_1) f_{1}(u)+(t/2)^{1/3}\xi_1 f_{2}(u)+f_3 (u)}},
\end{multline}
where
\begin{multline}\label{res}
\mathcal{I}_1 =(t/2)^{1/3}\frac{1}{2\pi \I}\oint_{\Gamma_{-1}}\dx u e^{t(f_{0}(\alpha-2-u)-f_{0}(u))}e^{(t/2)^{2/3}((a+u_2)f_{1}(\alpha-2-u)-(a+u_1)f_{1}(u))}\\
\times e^{(t/2)^{1/3}(\xi_2 f_{2}(\alpha-2-u)-\xi_1 f_2(u))}e^{-f_{3}(u)}.
\end{multline}

The contours $\Gamma_{0}$ and $\Gamma_{-1}$ in the double integral in \eqref{nores} satisfy \mbox{$\alpha-2-\Gamma_{-1}\supset \Gamma_{0}$} and $\Gamma_{-1}$ passes through the critical point $-1/2$. To provide the integrable bound for the double integral in \eqref{nores}, one does the same analysis as in the $\sigma_1 \geq \sigma_2$ case, except that the roles of $\Gamma_{-1}$ and $\Gamma_{0}$, and $\sigma_{1}$ and $\sigma_{2}$ are reversed: We modify $\Gamma_{0}$ by a vertical part with distance $(\sqrt{\sigma}_2+2\kappa/t^{1/3})/2$ to $-1/2$ and then go through the steps (a) to (d).

We have for $\Gamma_{-1}$ as in Figure~\ref{FigPaths1}, $\sigma_1=\sigma_2=0$ and $t$ large enough the bound
\begin{equation}\label{hypo}
\mathcal{I}_1 =\frac{(t/2)^{1/3}}{2\pi \I}\oint_{\Gamma_{-1}}|\dx u| \,\bigg|\frac{e^{t f_{0}(\alpha-2-u)}e^{(t/2)^{2/3}(a+u_2)f_{1}(\alpha-2-u)}}{e^{t f_{0}(u)}e^{(t/2)^{2/3}(a+u_1)f_{1}(u)}e^{f_{3}(u)}}\bigg|\leq c_2.
\end{equation}
The bound \eqref{hypo} follows from the identity in \eqref{nores} and the fact that respective bounds hold for $K^{2,{\rm resc}}$ and the double integral in \eqref{nores}.

For $\mathcal{I}_1$ we modify $\Gamma_{-1}$ near $-1/2$ by a vertical piece
\begin{equation}
\Gamma_{{\rm vert}}=\{-1/2-\sqrt{\sigma_2}(1+\I\rho)/2, \rho \in [-b,b]\}
\end{equation}
where $b> \tan(\pi/6)$.

Compared to $\sigma_1=\sigma_2=0$, the integrand has the additional factor
\begin{equation}
\exp(t\sigma_2 f_{2}(\alpha-2-u))\exp(-t\sigma_1 f_2(u))
\end{equation}
(e) We can choose the contour $\Gamma_{-1}$ such that $|1+u|<1/2$
for all $u\in\Gamma_{-1}$. In particular for $u\in\Gamma_{-1}\setminus\Gamma_{{\rm vert}}$ we may simply bound
\begin{equation}
|\exp(-t\sigma_1 f_2(u))|=\exp(t\sigma_1 \ln(2|1+u|)\leq 1.
\end{equation}
(f) Furthermore, $\Gamma_{-1}$ may be chosen such that for $u\in\Gamma_{-1}\setminus\Gamma_{{\rm vert}}$ the minimum of $|\alpha-1-u|=|u+2^{4/3}a/t^{1/3}|$ is reached at $\rho=\pm b$.
For this $u$ we have
\begin{equation*}
(2|u+2^{4/3}a/t^{1/3}|)^{2}=(1+2^{7/3}/t^{1/3}+\sqrt{\sigma_2})^{2}+\sigma_2 b^{2}\geq1+\sqrt{\sigma_2},
\end{equation*}
so that we get the bound for $L$ large and $\varepsilon$ small
\begin{equation*}
e^{t \sigma_{2} f_{2}(\alpha-2-u)}=e^{-t\sigma_2 \ln((2|u+2^{4/3}a/t^{1/3}|)^{2})/2}e^{\Or(t\sigma_{2}^{2})}\leq e^{-t\sigma_{2}^{3/2}/4}\leq e^{-\xi_{2}^{3/2}/6}\leq e^{-(\xi_1+\xi_2)}.
\end{equation*}
Now we deal with $\Gamma_{{\rm vert}}$. We write
\begin{equation}
\alpha-2-u=-1/2+V_1 \quad u=-1/2+V_2
\end{equation}
with $V_2=-\sqrt{\sigma_2}(1+\I\rho)/2$ and $V_1=-V_2-2^{4/3}/t^{1/3}$. Next we do Taylor around $-1/2$ in $f_0$ and we first obtain
\begin{equation}\label{mycal}
e^{t(f_{0}(\alpha-2-u)-f_{0}(u))}=e^{4t(V_1^{3}-V_{2}^{3})/3}e^{\Or(t\sigma_{2}^{2})}.
\end{equation}
We compute
\begin{equation}
\begin{aligned}\label{mytayl}
\Re(4t V_{1}^{3}/3)&=-2^{6}/3+2^{11/3}\sqrt{\sigma_2}	t^{1/3}-2^{4/3}\sigma_{2}t^{2/3}(1-\rho^{2})+t\sigma_{2}^{3/2}(1-3\rho^{2})/6.
\end{aligned}
\end{equation}
For $L$ large, the last term in \eqref{mytayl} dominates, thus \eqref{mytayl} reaches its maximum at $\rho=0$. So we may bound
\begin{equation}\label{final1}
|\eqref{mycal}|\leq e^{\xi_{2}^{3/2}(1/(3\sqrt{2})+c_4 \sqrt{\varepsilon})},
\end{equation}
with $c_4 >0$ a constant.

As for $f_1$, by Taylor expansion around $-1/2$ we obtain \begin{equation}\label{myf2}e^{(t/2)^{2/3}((a+u_2)f_{1}(\alpha-2-u)-(a+u_1)f_{1}(u))}=e^{(t/2)^{2/3}\sigma_2 (1+\I\rho)^{2}(u_2-u_1)}e^{\Or(\sqrt{\xi_2})}e^{\Or(t^{2/3}(V_{1}^{3}-V_{2}^{3}))}.
\end{equation}
Thus, for $L$ large we can bound
\begin{equation}\label{final2}
|\eqref{myf2}|\leq e^{c_5 \xi_2},
\end{equation}
for some constant $c_5>0$.

Finally, for $f_2$ we obtain
\begin{equation}
\begin{aligned}\label{myf3}
e^{(t/2)^{1/3}(\xi_2 f_{2}(\alpha-2-u)-\xi_1 f_2(u))}&=e^{t\sigma_{2}(-2V_1) }e^{t\sigma_{1}2V_2}e^{\Or(t\sigma_{2}V_{1}^{2})}e^{(t\sigma_{1}V_{2}^{2})}
\\&\leq e^{-t\sqrt{\sigma_{2}}(1+\I\rho)(\sigma_2+\sigma_1)}
e^{c_7\xi_2}e^{c_8\sqrt{\varepsilon}\xi_{2}^{3/2}},
\end{aligned}
\end{equation}
with $c_7>0$ and $c_8>0$ some constants.
Therefore, for $L$ large and $\varepsilon$ small enough we obtain
\begin{equation}\label{final3}
|\eqref{myf3}|\leq e^{-\xi_{2}^{3/2}/(2\sqrt{2})}
\end{equation}
Now \eqref{final3} dominates \eqref{final1}, \eqref{final2} for $L$ large and $\varepsilon$ small enough. So,
putting together (e), (f) and \eqref{hypo}, \eqref{final1}, \eqref{final2}, \eqref{final3} we obtain the desired bound for $|K^{2,{\rm resc}}|$.
\end{proof}
\begin{prop}[Large deviations for $K^1,K^2$]\label{PropLargeDev}
Let $\varepsilon>0$. Then, there is a finite constant $C$ such that for $t$ large enough we have
\begin{equation}
\left|K^{1,{\rm resc}}(u_1,\xi_1;u_2,\xi_2)+K^{2,{\rm resc}}(u_1,\xi_1;u_2,\xi_2)\right| \leq C e^{-(\xi_1+\xi_2)/2}.
\end{equation}
for $\xi_1,\xi_2 \geq \varepsilon t^{2/3}$.
\end{prop}
\begin{proof}
The estimate for $K_{1}^{{\rm resc}}$ is contained in Proposition~6 of ~\cite{BFS07}.
As in Proposition~\ref{PropModDev} we denote $\sigma_i=\xi_i 2^{-1/3}t^{-2/3}$ and distinguish the cases $\sigma_1 \geq \sigma_2$ and $\sigma_1 \leq \sigma_2$.

\textbf{Case} $\sigma_1 \geq \sigma_2$. We choose the same contours as in the moderate deviations regime for $\sigma_1 \geq \sigma_2$, here with $\sigma_1=\sigma_2=\varepsilon/2$ (the additional shift by $2\kappa/t^{1/3}$ in $\Gamma_{\mathrm{vert}}$ is however unnecessary for $t$ large enough). We write \mbox{$f_{0,\sigma}(v)=f_0(v)-\sigma\ln(2(1+v))$} so that
\begin{equation}
f_{0,\sigma}=f_{0,\varepsilon/2}(v)-(\sigma_2-\varepsilon/2)\ln(2(1+v)).
\end{equation}
Thus, compared to the $\sigma_1=\sigma_2=\varepsilon/2$ case we have the additional factor
\begin{equation}\label{addfac}
e^{-t(\sigma_2-\varepsilon/2)\ln(2(1+v))}e^{t(\sigma_1-\varepsilon/2)\ln(2(1+u))}.
\end{equation}
It suffices to bound $|\eqref{addfac}| $ because the integrand for $\sigma_1=\sigma_2=\varepsilon/2$ is (uniformely for $t$ large enough) bounded in $L^{1}$.
The choice of contours is such that $|1+v|$ reaches its minimum at $v=-1/2$ and $|1+u|$ its maximum at $u=-1/2-\sqrt{\varepsilon/2}/2$. Using further $\sigma_1-\varepsilon/2\geq \sigma_1 /2$, we may bound
\begin{equation}
|\eqref{addfac}|\leq e^{t\sigma_1 \ln(1-\sqrt{\varepsilon/2})/2}\leq e^{-c_9 t^{1/3}\xi_1}\leq e^{-(\xi_1+\xi_2)}.
\end{equation}
for some constant $c_9 > 0$.

\textbf{Case} $\sigma_1 \leq \sigma_2$. We again choose the same contours as in the moderate deviations regime for $\sigma_1 \leq \sigma_2$, with $\sigma_2=\sigma_1=\varepsilon/2$ (again the additional shift by $2\kappa/t^{1/3}$ is unnecessary for $t$ large enough). We again integrate out the residue at $v=\alpha-2-u$. In the double integral \eqref{nores}, with respect to $\sigma_1 = \sigma_2=\varepsilon/2$ we get the same additional factor, which can now be bounded
\begin{equation}
e^{-t(\sigma_2-\varepsilon/2)\ln(2(1+v))}e^{t(\sigma_1-\varepsilon/2)\ln(2(1+u))}\leq e^{-t(\sigma_2-\varepsilon/2)\ln(1+\sqrt{\varepsilon/2})}\leq e^{-(\xi_1+\xi_2)}.
\end{equation}
As for the residue \eqref{res}, compared to $\sigma_1 = \sigma_2=\varepsilon/2$ we have the additional term
\begin{equation}
e^{t(\sigma_1-\varepsilon/2)\ln(2(1+u))}e^{-t(\sigma_2-\varepsilon/2)\ln(2(-u-2^{4/3}/t^{1/3}))}\leq e^{-(\xi_1+\xi_2)},
\end{equation}
where the inequality holds since $|2(-u-2^{4/3}/t^{1/3})|\geq 1+\sqrt{\varepsilon/2}$ and \mbox{$|2(1+u)|\leq 1$}.
\end{proof}

\newpage

\appendix

\section{Kernel $K_a$ in terms of Airy functions}\label{SectAppendix}
Here we give the explicit form of $K_a$ that we used for the numerical evaluation of $G_a$ and its statistics.
\begin{lem}
Denote $u_{i,a}=u_i+a$ we have (with the conjugation transferred to the diffusion part)
\begin{align}
&K_a (u_1,\xi_1;u_2,\xi_2)\stackrel{\rm conj}{=}-\frac{e^{\frac{2}{3}u_{1,a}^{3}+u_{1,a}\xi_1}}{e^{\frac{2}{3}u_{2,a}^{3}+u_{2,a}\xi_2}} \frac{e^{-(\xi_2-\xi_1)^{2}/(4(u_2-u_1))}}{\sqrt{4\pi(u_2-u_1)}}\Id(u_2>u_1)\label{diff} \\
&+\int_{0}^{\infty}\dx \lambda \Ai(\xi_1+u_{1,a}^{2}+\lambda)\Ai(\xi_2+u_{2,a}^2+\lambda)e^{\lambda(u_2-u_1)}\label{vanishing} \\
& +\int_{0}^{\infty}\dx\lambda \Ai(\xi_1+u_{1,a}^{2}-\lambda)\Ai(\xi_2+u_{2,a}^2+\lambda)e^{\lambda (2a+u_1+u_2)}\label{div1} \\
&-\int_{0}^{\infty}\dx\lambda \Ai(\xi_1+u_{1,a}^{2}+\lambda)\Ai(\xi_2+u_{2,a}^2+\lambda)e^{\lambda (4a+u_2-u_1)} \label{divstrong}\\
&+\int_{0}^{\infty}\dx\lambda \Ai(\xi_1+u_{1,a}^{2}+\lambda)\Ai(\xi_2+u_{2,a}^2-\lambda)e^{\lambda (2a-u_1-u_2)}. \label{div2}
\end{align}
\end{lem}
\begin{proof}
The result is an easy computation that uses the identities
\begin{equation}
\begin{aligned}
&\frac{-1}{2\pi \I}\int_{\delta+\I\R}\dx v e^{v^{3}/3+xv^{2}+yv}=\Ai(x^{2}-y)e^{\frac{2}{3}x^{3}-xy},
\\&\frac{1}{z}=\int_{0}^{\infty}\dx \lambda e^{-\lambda z} \qquad (z \in\mathbb{C},\,\Re(z)>0),
\end{aligned}
\end{equation}
for any $\delta>\max\{0,x\}$.
\end{proof}
\begin{remark}
Alternatively, via the identity (A.6) of~\cite{BFS07}, one has
\begin{equation}
\begin{aligned}\label{alter}
\eqref{div1}&=-\int_{-\infty}^{0}\dx \lambda e^{\lambda(u_{2,a}+u_{1,a})}\Ai(\xi_1+u_{1,a}^{2}-\lambda)\Ai(\xi_2+u_{2,a}^{2}+\lambda)
\\&+2^{-1/3}\Ai\left(2^{-1/3}(\xi_1+\xi_2)+2^{-4/3}(u_1-u_2)^{2}\right)e^{-\frac{1}{2}(u_{1,a}+u_{2,a})(\xi_2+u_{2,a}^{2}-\xi_1-u_{1,a}^{2})},
\end{aligned}
\end{equation}
with an analogous formula for \eqref{div2}.
\end{remark}


\end{document}